\documentclass{llncs}
\usepackage{makeidx}  
\usepackage{latexsym} 
\usepackage{amsmath}
\usepackage{amssymb}
\usepackage[pdftex]{graphics} 
\usepackage[dvips]{graphicx} 
\usepackage{xypic} 
\xyoption{all} 
\input{Qcircuit}

\newcommand{\C}{{\mathcal C}}

\usepackage{bbold}

\newtheorem{notation}[theorem]{Notation}{\bfseries}{\itshape}

\begin{document}

\begin{frontmatter}
\title{Quantum Speedup and Categorical Distributivity}
\author{Peter Hines}
\institute{University of York \\ 
\email{peter.hines@york.ac.uk}}
\end{frontmatter}

\maketitle 
\noindent

\begin{abstract}This paper studies one of the best-known quantum algorithms --- Shor's factorisation algorithm --- via categorical distributivity. A key aim of the paper is to provide a minimal set of categorical requirements for key parts of the algorithm, in order to establish the most general setting in which the required operations may be performed efficiently.

We demonstrate that Laplaza's theory of coherence for distributivity \cite{ML1,ML2} provides a purely categorical proof of the operational equivalence of two quantum circuits, with the notable property that one is exponentially more efficient than the other. This equivalence also exists in a wide range of categories.  

When applied to the category of finite-dimensional Hilbert spaces, we recover the usual efficient implementation of the quantum oracles at the heart of both Shor's algorithm and quantum period-finding generally; however, it is also applicable in a much wider range of settings.
\end{abstract}



\section{Introduction}
\subsection{Shor's Algorithm: Oracles and Quantum Fourier Transforms}
The structure of Shor's algorithm is deceptively simple: an oracle which acts classically on the computational basis computes modular exponentials; this oracle is conjugated by the circuit for the quantum Fourier transform. Up to some relatively simple classical post-processing (computing continued fraction expansions), this is enough to find the prime factors of a number in an exponentially fast time -- at least, as compared with the {\em best known} classical algorithm.


The traditional view of Shor's algorithm and other quantum period-finding algorithms is that their power arises from the quantum Fourier transform; \cite{NC} lists Shor's algorithm in the section ``Applications of the Fourier transform''. This was challenged in \cite{ALM}, where it was demonstrated that the quantum Fourier transform has a {\em low bubble width} circuit --- and any quantum algorithm that is built entirely from low bubble-width circuits has an efficient classical simulation. Thus, it appears that the obstacle to an efficient classical simulation of Shor's algorithm is the oracle for modular exponentiation, rather than the conjugating quantum Fourier transform. 
 
 The conclusions drawn in \cite{ALM} (the claim that the quantum power of Shor's algorithm arises from the central oracle) were, and remain, controversial. However, further evidence to support this claim was provided in \cite{YS}, where it was demonstrated that modular exponentiation, in of itself, is sufficient.  From \cite{YS}: {\em Any classical algorithm that can efficiently simulate the circuit implementing modular exponentiation for general product input states and product state measurements on the output, allows for an efficient simulation of the entire Shor algorithm on a classical computer.} A special case  of this, as noted in \cite{YS},  would be any tensor contraction scheme for the modular exponentiation circuit.

\subsection{The Aims of this Paper}
This paper describes the circuit for modular exponentiation used in \cite{SH} in purely categorical terms. The motivation is to find the most general structures in which this precise form of the oracle may be implemented. We therefore avoid, where possible, categorical machinery that is closely or uniquely associated with the theory of finite-dimensional Hilbert spaces.\footnote{In particular, the constructions we will present are significantly simpler in the presence of compact closure and biproducts -- two categorical properties closely associated with quantum mechanics.  However, neither of these categorical properties are necessary, so we work in the  more general setting.} Instead, we will simply require a category with two monoidal tensors related by a notion of distributivity.  As this is established for abstract categories, any concrete category satisfying this simple requirement is sufficient.

\subsection{The Structure of the Paper}
This paper is divided into two sections: pure category theory, and concrete realisations of this abstract theory. 
\begin{enumerate}
\item We first use the abstract theory of categories with two monoidal tensors related by distributivity to define endofunctors and further categorical operations on such categories. We use these to define an `iterator' operation $!^N(\ )$ on endomorphism monoids of such categories, and 
use Laplaza's theory of coherence for distributivity to give an exponentially efficient factorisation of this operation.
\item The second half of the paper gives a concrete realisation of this operation, and its efficient factorisation, within the quantum circuit paradigm. The $!^N(\ )$ operation has a concrete realisation as the oracle required for quantum period-finding, and its efficient factorisation  is exactly Shor's implementation of modular exponentiation oracle.
\end{enumerate}

\section{Basic Definitions}
Our abstract setting is that of categories with distributivity, defined in \cite{ML1,ML2}:
\begin{definition}\label{basics}
A {\bf category with distributivity} is a category $\mathcal C$  with two distinct symmetric monoidal tensors: the {\bf multiplicative tensor} 
$(\  \otimes \ ) : \C \times \C \rightarrow \C$ 
and the {\bf additive tensor} 
$(\  \oplus \ ) : \C \times \C \rightarrow \C$
 that are related by natural distributivity monomorphisms
\begin{eqnarray}
& & dl_{ABC}: A\otimes (B\oplus C) \rightarrow (A\otimes B)\oplus (A\otimes C) \label{leftdistrib} \\
& & dr_{XYZ}: (X\oplus Y)\otimes Z \rightarrow (X\otimes Z) \oplus (Y\otimes Z) \label{rightdistrib} 
\end{eqnarray}
 satisfying coherence conditions 
laid out in \cite{ML1,ML2}. 
\end{definition}
The required coherence conditions are decidedly non-trivial and form an infinite family of diagrams that are required to commute, although these may be significantly simplified (from \cite{ML1}, {\em ``we are reduced to a finite number of types of diagrams if we drop unnecessary commutativity conditions''}).

\begin{notation}
We adopt the convention of using the Greek alphabet for the structural isomorphisms related to the multiplicative tensor, and the Roman alphabet for the additive tensor. 
We denote the multiplicative associativity and symmetry isomorphisms by $ \tau_{XYZ} : X\otimes (Y\otimes Z) \rightarrow (X\otimes Y)\otimes Z$ and $\sigma_{X,Y}:X\otimes Y \rightarrow Y\otimes X$, and the additive associativity and symmetry by $ t_{XYZ}:X\oplus (Y\oplus Z)\rightarrow (X\oplus Y)\oplus Z$ and $ s_{XY}:X\oplus Y \rightarrow Y\oplus X$. We will frequently appeal to MacLane's coherence theorem for associativity, and treat both the multiplicative and additive tensors as strict.

We will also denote the multiplicative unit object by $I$, and the additive unit object by $0$. 
\end{notation}

A special case that is often considered (e.g. \cite{CLW,RC}) is where the distributivity monomorphisms are in fact isomorphisms. 
\begin{definition}
Let $(\C,\otimes ,\oplus)$ be a category with distributivity. We say that is is {\bf strongly distributive} when the natural distributivity monomorphisms have global inverses, 
\begin{eqnarray}
& & dl^{-1}_{ABC}: (A\otimes B)\oplus (A\otimes C)  \rightarrow A\otimes (B\oplus C)\label{leftdistrib} \\
& & dr^{-1}_{XYZ}: (X\otimes Z) \oplus (Y\otimes Z) \rightarrow (X\oplus Y)\otimes Z \label{rightdistrib} 
\end{eqnarray}
\end{definition}

Strongly distributive categories are a special case of Definition \ref{basics}, so we may still appeal to Laplaza's coherence theorems. Appropriate care will be taken when using commutative diagrams containing inverses of these canonical isomorphisms to ensure that an equivalent result may be derived without the use of inverses. See the proof of Lemma \ref{invexample} for an example of this.  


\subsection{Distinguished Objects, and Copying Functors}
Strongly distributive categories have two distinguished objects: the additive and multiplicative unit objects $0,I\in Ob(\C)$. Their interaction with the two monoidal tensors is given (up to straightforward canonical isomorphism) by the following tables:

\begin{center}
\begin{tabular}{c|ccc}
$\otimes$ & $0$ & \ \ & $I$ \\
\hline 
$0$ & $0$ & &  $0$ \\
\ 	& \	& & \ \\
$I$ & $0$ & & $I$ 
\end{tabular}
\hspace{6em}
\begin{tabular}{c|ccc}
$\oplus$ & $0$ & & $I$ \\
\hline 
$0$ & $0$ & & $I$ \\
	& 	& & \\
$I$ & $I$ & & $I\oplus I$ 
\end{tabular}
\end{center}

Observe that $I\oplus I$ is neither $0$ nor $I$; thus in the absence of any further identities,  
strongly distributive categories have additional distinguished objects. 
\begin{definition}
We define ${\mathbb 2} \in Ob(\mathcal C)$ to be the additive tensor of two multiplicative units, so ${\mathbb 2} = (I\oplus I)$.
\end{definition}
Such objects are considered in \cite{RC}, where -- in the special case that $\otimes$ and $\oplus$ are a product and coproduct respectively -- they generate Boolean algebras. The classical logical interpretation is well-established. As noted in \cite{BCST} the form of  distributivity introduced in \cite{ML1} is entirely unsuitable for linear logic,  since distributivity implies a form of  `copying' operation that we now describe:

\begin{lemma}\label{invexample}
Let $(\C,\otimes ,\oplus)$ be a strongly distributive category. Then
\begin{enumerate}
\item $\mathbb 2 \otimes X \cong X\oplus X$
\item for all $f\in \C(X,Y)$, the following diagram commutes:
\[ \xymatrix{ 
		\mathbb 2 \otimes X \ar[rr]^{1_{\mathbb 2} \otimes f} &  & \mathbb 2 \otimes Y \ar[d]^{dr_{I,I,Y}} \\
X\oplus X \ar[u]^{dr^{-1}_{I,I,X}} \ar[rr]^{f\oplus f} & & Y\oplus Y }
\]
\end{enumerate}
\end{lemma}
\begin{proof}
$ $ \\
\begin{enumerate} 
\item By distributivity, and the fact that $I$ is the unit object for the multiplicative tensor, $\mathbb 2 \otimes A = (I\oplus I )\otimes A \cong (I\otimes A)\oplus (I\otimes A)  \cong A\oplus A$.
\item By naturality, 
the following diagram commutes:
\[ \xymatrix{
(I\oplus I) \otimes X \ar[rr]^{1_{I\oplus I}\otimes f} \ar[d]|{dr_{I,I,X}} & &(I\oplus I)\otimes Y \ar[d]|{dr_{I,I,Y}} \\ 
 X\oplus X \ar[rr]_{f\oplus f} & & Y\oplus Y } \]
As $(\C,\otimes ,\oplus)$ is strongly distributive, we may replace $dr_{IIX}$ in the above diagram by $dr^{-1}_{IIX}$, and reverse the corresponding arrow:
\[
\xymatrix{
(I\oplus I) \otimes X \ar[rr]^{1_{I\oplus I}\otimes f} & & (I\oplus I)\otimes Y \ar[d]|{dr_{I,I,Y}} \\ 
 X\oplus X \ar[rr]_{f\oplus f}  \ar[u]|{dr^{-1}_{I,I,X}}& & Y\oplus Y } 
 \]
\end{enumerate}
\qed \end{proof}


\begin{definition}\label{times2}
Let $(\C ,\otimes ,\oplus)$ be strongly distributive. We define the {\bf copying endofunctor} to be  $\delta = (\mathbb 2 \otimes \underline{\ }):\C \rightarrow \C$.
\end{definition}

This terminology is motivated by the following result: 
\begin{proposition}\label{deltacopying}
Let $\Delta:\C\rightarrow \C\times \C $ be the diagonal functor given by
\begin{itemize}
\item {\bf (Objects)} $\Delta (A) = (A,A)$.
\item {\bf (Arrows)} $\Delta(f)=(f,f)$
\end{itemize}
Then there exists a natural isomorphism (i.e. a natural transformation whose components are isomorphisms) from the composite functor 
$(\underline{\ } \oplus \underline{ \ }) \Delta :\C \rightarrow \C$ to the functor $(\mathbb 2 \otimes \underline{ \ }): \C\rightarrow \C$.

We draw this diagrammatically, as follows:
\[ \xymatrix{ 
								& \C \times \C \ar@{=>}[d] 	  \ar[dr]^{(\underline{ \ }\oplus \underline{\ })} &  \\
		\C	\ar[ur]^{\Delta}	\ar[rr]_{(\mathbb 2 \otimes \underline{\ })}		& 			& \C 
}				\]
\end{proposition}

\begin{proof}
For arbitrary $X\in Ob(\C)$, the components of this natural transformation are given by the distributivity isomorphisms $dl_{I,I,X}:\mathbb 2 \otimes X \rightarrow X\oplus X$ (treating units arrows as strict). The required identity then follows from Lemma \ref{invexample}.
\qed \end{proof}

\begin{remark}
At first sight, this `copying' behaviour appears to be at odds with the `no-cloning'  and `no-deleting' theorems \cite{WZ,PB} of quantum information.  However, these are based on tensor products (`multiplicative' tensors), whereas the functor of Definition \ref{times2} acts as a form of copying for the additive structure -- it is related to the fanout operation \cite{HS} rather than the forbidden quantum cloning.
\end{remark}


Iterating a copying operation gives a form of exponential growth, as we demonstrate:

\begin{corollary}
For all $f\in \C(X,Y)$ and $n\geq 1\in \mathbb N$, there exists canonical isomorphisms $\lambda^{(n)}_X:\mathbb 2^{\otimes n}\otimes X \rightarrow \bigoplus_{j=0}^{2^n-1} X$ making the following diagram commute:
\[ \xymatrix
	{
 		{\mathbb 2}^{\otimes n}\otimes X \ar[rr]^{\delta^n (f)} \ar[d]_{\lambda_X} & & {\mathbb 2}^{\otimes n} \otimes Y \\
		\bigoplus_{j=0}^{2^n-1} X \ar[rr]_{\oplus_{n=0}^{2^n-1} f} & & \bigoplus_{j=0}^{2^n-1} Y\ar[u] _{\lambda_Y^{-1}}
	}
\]
\end{corollary}
\begin{proof}
We give the canonical isomorphisms $\lambda_X^{(n)}:\mathbb 2^{\otimes n}\otimes X \rightarrow \bigoplus_{j=0}^{2^n-1} X$ by induction: we  take $\lambda_X^{(1)}=dr_{I,I,X} : \mathbb 2 \otimes  X \rightarrow X\oplus X$, and 
\[ \lambda_X^{(n)} =dr_{I,I,\oplus_{j=0}^{2^{n-1}-1}X}\left(1_{\mathbb 2}\otimes \lambda^{(n-1)}_X\right) .\] 
The above diagram then commutes by naturality. 
\qed
\end{proof}

We now demonstrate that $\delta : \C \rightarrow \C$ is a (weak) {\em monoidal} endofunctor, for the additive, but not multiplicative, structure. 


\begin{proposition} The functor $\delta =(\mathbb 2 \otimes \underline{\ \ } ) : \C \rightarrow \C$ does not preserve the multiplicative monoidal structure, even up to isomorphism; however the additive structure is preserved up to a simple distributivity isomorphism. 
\end{proposition}
\begin{proof}
To see that $\delta$ does not preserve the multiplicative tensor, observe that note that
\[ \delta(A)\otimes \delta(B) \ = \ (1_\mathbb 2 \otimes \sigma_{A,\mathbb 2}\otimes 1_B) \delta^2(A\otimes B) \]
Thus, unless $\delta(X)\cong \delta^2(X)$ for all $X\in Ob(\C)$, the copying functor does not preserve the multiplicative tensor, even up to isomorphism.

However, $\delta(0)\cong 0$, and the following diagram also commutes:
\[ \xymatrix{
\delta(A\oplus B) 	\ar@{-}[rr]^\cong  \ar@{-}[d]_{=}  &  &\delta(A) \oplus \delta(B) \ar@{-}[d]^{=}  \\
\mathbb 2 \otimes (A\oplus B) \ar[rr]_{dl_{(I\oplus I),A,B}}  & & \mathbb 2 \otimes A \oplus \mathbb 2 \otimes B \\
}
\]
Since the required isomorphisms are canonical coherence isomorphisms in both cases, $\delta:(\C,\oplus )\rightarrow (\C,\oplus )$ is a (weak) monoidal functor.  \qed \end{proof}

\subsection{Copying and the Iterator}\label{bang}
We now study an operation on endomorphism monoids closely related to the copying functor $\delta:(\C,\oplus) \rightarrow (\C,\oplus)$. 

\begin{definition}
Let $(\C,\otimes ,\oplus)$ be a strongly distributive category. For all $f\in \C(A,A)$, we define the {\bf $\bf N^{th}$ iterator of $f$} to be
\[ !^N(f) \ = \ \bigoplus_{j=0}^{N-1}f^j \ \in \ \C \left( A^{\oplus N} , A^{\oplus N}\right) \]
\end{definition}

We will give an efficient factorisation of $!^{2^n}(f)$. This will rely on the following interaction of the functor $\delta=(\mathbb 2 \otimes \underline{\ \ }):\C\rightarrow \C$, and the  multiplicative and additive symmetries $\sigma_{X,Y}:X\otimes Y \rightarrow Y\otimes X$ and $s_{A,B}:A\oplus B \rightarrow B \oplus A$. 

\begin{lemma}\label{deltasym}
Let $A,B,C$ be objects of a strongly distributive category $(\C,\otimes ,\oplus)$. Then the following diagram commutes:\\

\scalebox{0.8}{
\xymatrix {
																					&	\mathbb 2 \otimes A \otimes (B\oplus C) \ar[dl]|{\sigma_{\mathbb 2 ,A} \otimes 1_{B\oplus C}} 	\ar[dr]|{1_{\mathbb 2} \otimes dl_{A,B,C}} 	&				\\
A\otimes \mathbb 2 \otimes (B\oplus C) 	\ar[d]|{1_A\otimes dl_{\mathbb 2 , B, C}}						&																											& \mathbb 2 \otimes (A\otimes B\oplus A \otimes C) \ar[d]|{dr_{\mathbb 2,A\otimes B,A\otimes C}} \\
A\otimes (\mathbb 2 \otimes B \oplus \mathbb 2 \otimes C) \ar[d]|{1_A\otimes (dr_{I,I,B}\oplus dr_{I,I,B})}		& 																											& A\otimes B \oplus A \otimes C \oplus A \otimes B \oplus A \otimes C  \ar[d]|{1_{A\otimes B}\oplus s_{A\otimes C,A\otimes B}\oplus 1_{A\otimes C}}  \\
A\otimes (B\oplus B \oplus C \oplus C ) \ar[rr]|{dl_{A,B\oplus B ,C\oplus C}}							&																											& A\otimes B \oplus A \otimes B \oplus A \otimes C \oplus A \otimes C \\
}
}
\end{lemma}
\begin{proof}The commutativity of this diagram follows immediately from the coherence theorems of \cite{ML1,ML2} (note that we have elided associativity isomorphisms, for clarity).  
\qed \end{proof}

\begin{theorem}\label{maintheorem}
For arbitrary $n\geq 1$ and $f\in \C(X,X)$, the arrow $!^{2^{n+1}}(f)$ can be defined in terms of $!^{2^n}(f)$, the functor $(\mathbb 2 \otimes \underline{\ \ })$, and canonical isomorphisms, with the exact relationship expressed by the commutativity of the following diagram:\\

\scalebox{0.7}{
\xymatrix
 {
 	\mathbb 2 \otimes \bigoplus_{j=1}^{2^{n-1}} X 	\ar[dr]|{1_{\mathbb 2}\otimes dl^{-1}}								&	& 		& 	& \mathbb 2\otimes  \bigoplus_{j=1}^{2^{n-1}} X 	\ar[llll]_{1_{\mathbb 2}\otimes !^{2^{n-1}}(f)}						\\
	 																							& \mathbb 2 \otimes \mathbb 2^{\otimes (n-1)} \otimes X \ar[d]|{\sigma_{\mathbb 2,\mathbb 2^{\otimes (n-1)}}}	& 		&	 \bigoplus_{j=1}^{2^n} X	\ar[ur]|{dl^{-1}_{\mathbb 2 ,\underline{\ \ }}} \ar[d]|{!^{2^n}(f)}&	\\
																								& \mathbb 2^{\otimes (n-1)}\otimes \mathbb 2  \otimes X \ar[dl]|{1_{\mathbb 2^{\otimes (n-1)}}\otimes dr_{I,I,X}}		& 		&	 \bigoplus_{j=1}^{2^n} X	&		\\			
	\mathbb 2^{\otimes (n-1)}\otimes (X\oplus X)	\ar[rrrr]_{1_{\mathbb 2^{\otimes(n-1)}}\otimes \left(1_X\oplus f^{2^{n-1}}\right)}	&	& 		& 	& \mathbb 2^{\otimes (n-1)}\otimes (X\oplus X)		\ar[ul]|{dl_{\mathbb 2 ,\underline{\ \ }}}	 	\\		
	}
}

\end{theorem}
\begin{figure}[]
\caption{A technical result implied by naturality}\label{keynatresult}
\scalebox{0.7}{
\xymatrix
 {
		\mathbb 2 \otimes A \otimes (B \oplus C)\ar[rrrr]^{1_{\mathbb 2} \otimes 1_A \otimes (f\oplus g)} \ar[dd]|{\sigma_{\mathbb 2,A}\otimes 1_{B\oplus C}} 	& 	&&	& \mathbb 2 \otimes A \otimes (Y\oplus Z) \ar[dd]|{\sigma_{\mathbb 2,A}\otimes 1_{Y\oplus Z}} \\
																															&	&&	&							\\
		A\otimes \mathbb 2 \otimes (B\oplus C)\ar[rrrr]^{1_A\otimes 1_{\mathbb 2} \otimes (f\oplus g)}	\ar[dd]|{1_A\otimes dl_{\mathbb 2 ,B,C}}				&	&&	&	A\otimes \mathbb 2 \otimes (Y\oplus Z) \ar[dd]|{1_A\otimes dl_{\mathbb 2 ,Y,Z}}	 \\
																															&	&&	&						\\		
		A\otimes (\mathbb 2 \otimes B \oplus \mathbb 2 \oplus C) \ar[dd]|{1_A\otimes (dr_{I,I,B}\oplus dr_{I,I,C})}	\ar[rrrr]^{1_A\otimes (1_{\mathbb 2}\otimes f \oplus 1_{\mathbb 2}\otimes g)}					&	&&	&	A\otimes (\mathbb 2 \otimes Y \oplus \mathbb 2 \otimes Z) \ar[dd]|{1_A\otimes (dr_{I,I,Y}\oplus dr_{I,I,Z})}		\\
																															&	&&	&						\\		
A\otimes (B \oplus B \oplus C \oplus C)\ar[dd]|{dl_{A,B\oplus B,C\oplus C}} \ar[rrrr]^{1_A\otimes (f\oplus f \oplus g \oplus g)}								&	&&	&	A\otimes (Y \oplus Y \oplus Z \oplus Z)\ar[dd]|{dl_{A,Y\oplus Y,Z\oplus Z}}		\\
																															&	&&	&						\\		
A\otimes B \oplus A \otimes B \oplus A \otimes C \oplus A \otimes C 	\ar[rrrr]^{1_A\otimes f \oplus 1_A\otimes f \oplus 1_A\otimes g \oplus 1_A \otimes g}			&	&&	&	A\otimes Y \oplus A \otimes Y \oplus A \otimes Z \oplus A \otimes Z	\\ 
}
}
\end{figure}
\begin{proof}
Consider the left hand path in the coherent diagram of Lemma \ref{deltasym} above, from $\mathbb 2 \otimes A \otimes (B\oplus C)$ to $A\otimes B \oplus A \otimes B \oplus A \otimes C \oplus A \otimes C$, along with arrows $f\in \C(B,Y)$ and $g\in \C(C,Z)$. Then naturality of canonical coherence isomorphisms implies the commutativity of the diagram in Figure \ref{keynatresult}.
The required result is then the special case where $X=Y$ and $A=\mathbb 2^{\otimes n}$.
\qed \end{proof}

\subsection{String Diagrams for Categories with Distributivity}\label{strings}
Results such as Theorem \ref{maintheorem} above may be given as string diagrams, using the conventions formalised in \cite{GTC1,GTC2}.  When we have two distinct monoidal tensors, we adopt various conventions to ensure that such a diagrammatic reasoning is still valid:
\begin{enumerate}
\item Lines are separated by an implicit  {\em multiplicative} rather than an {\em additive} monoidal tensor.
\item Operations involving {\em additive} monoidal tensors are enclosed in a double box.
\item Entering / leaving a double box requires an (implicit) distributivity isomorphism / its inverse. 
Provided care is taken with labelling of objects, the required canonical isomorphism may be deduced from the type of the operation.
\end{enumerate}
We may then use diagrammatic manipulations on either the diagram as a whole (treating each additive box as a single operation), or on the contents of an individual double box (treating it as an entire diagram in of itself). These conventions ensure that the diagrammatic manipulations of \cite{GTC1,GTC2} are valid, simply by restricting the permitted manipulations.

Using the above, an illustration of Theorem \ref{maintheorem} is given in Figure \ref{secondfigure}.

\begin{corollary}\label{efficientconstr}
There exists an efficient construction of $!^{2^n}(f)$ in $O(n)$ steps, based on canonical coherence isomorphisms.
\end{corollary}
\begin{proof}
This follows by iterating the construction of Theorem \ref{maintheorem}. A diagrammatic illustration is given in Figure \ref{Shorsdiagram}.
\qed \end{proof}

\begin{figure}
 \caption{A `string diagram' illustration of Theorem \ref{maintheorem}}\label{secondfigure}
\begin{center}

\includegraphics{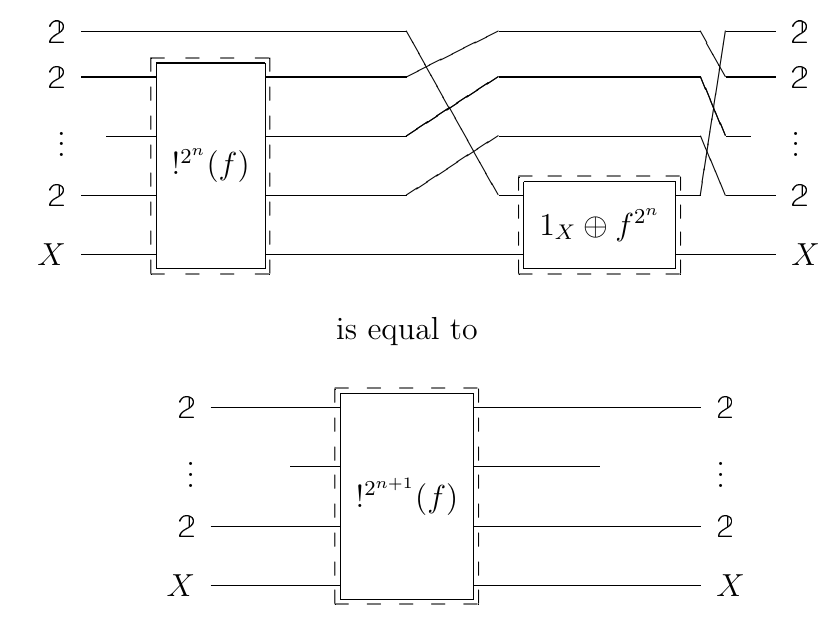}
\end{center}
\end{figure}

 \begin{figure}[]
 \caption{The efficient construction of $!^{16}(f)= 1_X\oplus f \oplus \ldots \oplus f^{15}\oplus f^{16}$}\label{Shorsdiagram}
\begin{center}
\includegraphics{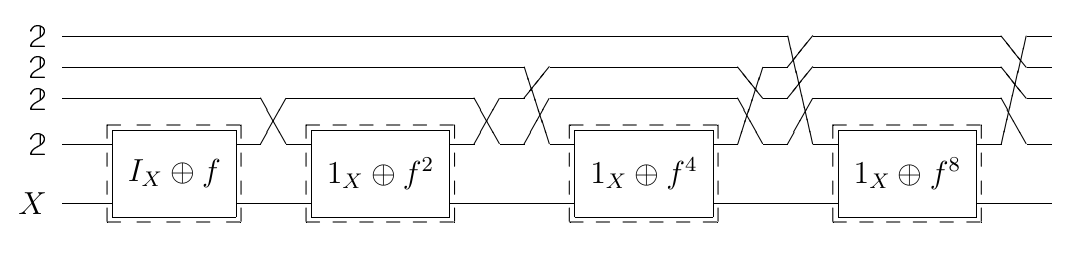}
\end{center}
\end{figure}

\section{Concrete Realisation in Hilbert Space}\label{real}
{\em The following sections assume a small degree of  familiarity with quantum circuits and Hilbert spaces. More details may be found in in \cite{NC} or any other text on quantum computing and information.}

We now consider the constructions of the previous section 
in the concrete setting of finite-dimensional complex Hilbert spaces. The two monoidal tensors are the familiar tensor product and direct sum --- the distributivity  isomorphisms relating the these are well-established. 

The translation of the basic concepts is given in Table \ref{concretehilb}. A subtlety of this is the interpretation of the distinguished object $\mathbb 2$. The direct sum of two $1$-dimensional spaces is of course a two-dimensional space. However, since $\mathcal Q$ is built up in this way, we should think of it as having a fixed orthonormal basis specified by the canonical inclusions\footnote{This is, of course, related to the  `classical structures' of \cite{CP} --- these are a special form of Frobenius algebra that play the role of orthonormal bases in categorical quantum mechanics. They are based on a `copying' operation; the connection between these, and the $\mathbb 2 \otimes \underline{\ \ }$ copying functor of Definition \ref{times2}, is straightforward.} -- this will allow us to use matrix representations for arrows in this category.

\begin{table}[t]
\caption{Translating abstract theory into a concrete setting}
\label{concretehilb}

\begin{center}
\begin{tabular}{|c|c|}
\hline
{\bf Abstract category $\C$} & {\bf Concrete category $\bf Hilb_{FD}$} \\
\hline \hline 
{\bf Multiplicative tensor} & {\bf Tensor product} \\
$H = H_1\otimes H_2$		& {\em (Treating two systems as a }  \\
				& {\em single compound system).} \\
				&						\\
{\bf Additive tensor} & {\bf Direct sum} \\
	$(U\oplus 1)$		& {\em ($U$ controlled on $\ket{0}$)}\\
	$(1\oplus V)$		& {\em ($V$ controlled on $\ket{1}$)}\\
					& 							\\
{\bf Multiplicative unit} $I$ & {\bf Complex plane} $\mathbb C$ \\
					&						\\
{\bf Additive unit} $0$ & {\bf Zero-dim. space} $\{ 0 \}$ \\
					&					\\	
{\bf Distinguished object} $\mathbb 2 = I \oplus I$ & {\bf Qubit space} $\mathcal Q$,  with \\
 								&  \hspace{1em} orthonormal basis $\{ \ket{0},\ket{1}\}$  \\
\hline
\end{tabular}
\end{center}

\end{table}


\begin{remark} A key point of this paper is that the structures required for the central oracle of Shor's algorithm  are {\em not} dependent on the machinery of either traditional quantum mechanics (such as a matrix calculus, or notions of linearity and convergence), or categorical reinterpretations (compact closure, biproduct structures, \&c.).  However, the existence of matrix representations certainly makes the concrete instantiation simpler, as the following sections will demonstrate.
\end{remark}

\subsection{Interpreting the Direct Sum in the Circuit Model}\label{oplusinterp}
In the translation from an abstract to a concrete setting provided in Table \ref{concretehilb}, the interpretation of the tensor, the multiplicative unit, and the distinguished object $\mathbb 2$ are standard. Furthermore, as we are forced by the category theory to specify an orthonormal basis for the two-dimensional qubit space, we are now, for all practical purposes, working within the quantum circuit paradigm. The final connection arises from the interpretation of the direct sum in terms of `quantum conditionals', or `controlled operations' \cite{TCS2}.

\begin{definition}
Let $U,V$ be unitary operations on a finite-dimensional Hilbert space $H$. The {\bf controlled operations} $Ctrl_0U$ and $Ctrl_1V$ 
are the operations on $\mathcal Q \otimes H$ defined by:
 \begin{center}
\begin{tabular}{ccc}
$Ctrl_0U\ket{0}\ket{\psi} =  \ket{0}U\ket{\psi}$ & and & $Ctrl_0U \ket{1}\ket{\psi} = \ket{1}\ket{\psi}$ \\
	& 	\\
$Ctrl_1V\ket{0}\ket{\psi} =  \ket{0}\ket{\psi}$ &  and & $Ctrl_1V \ket{1}\ket{\psi} = \ket{1}V\ket{\psi}$ \\
\end{tabular}
\end{center}
with  standard circuit representations shown in figure \ref{ctrlops}.
\end{definition}
Denoting the $n$-qubit identity operation by $I_n$, these operations have matrix representations given by $C_0U\ = \ \left( \begin{array}{cc} U & {\bf 0} \\ {\bf 0} & I_n \end{array}\right)$ and $ 
C_1V\ = \ \left( \begin{array}{cc} I_n & {\bf 0} \\ {\bf 0} & V \end{array}\right)$.
 \begin{figure}
 \caption{Quantum circuits for `Control on $0$' and  `Control on $1$'}\label{ctrlops}
\begin{center}
\scalebox{0.8}{
\includegraphics{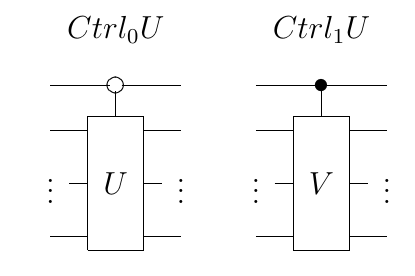}}
\end{center}
\end{figure}
The direct sum $U \oplus V$ 
is then simply the composite $U\oplus V = Ctrl_0U.Ctrl_1V = Ctrl_1V.Ctrl_0U$.
 
 
Controlled operations can themselves be controlled. Given $2^n$ unitary maps $\{ U_a\}_{a=0}^{2^n-1}$, the construction of $\bigoplus_{a=0}^{2^n-1}U_a$ is immediate via an $n$-qubit ancilla. This is illustrated in Figure \ref{multidirectsum}  for the direct sum of $2^3$ unitaries, and  the `binary counting' pattern on controls is immediate.
 \begin{figure}[]
 \caption{A circuit for the direct sum  $\bigoplus_{a=0}^{7} U_a$}\label{multidirectsum}
\begin{center}
\scalebox{0.8}{
\includegraphics{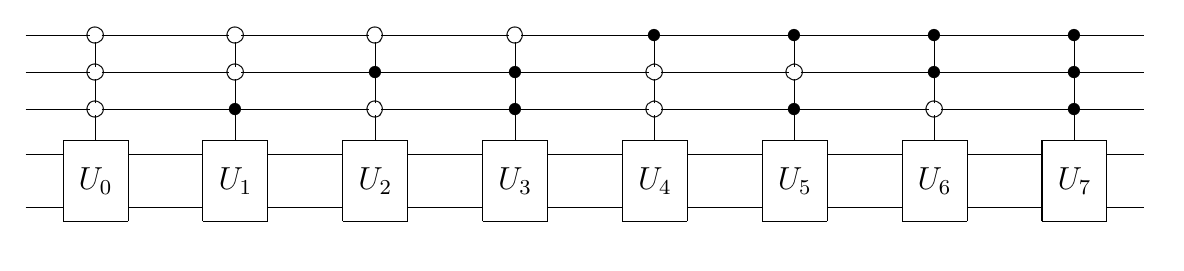}}
\end{center}
\end{figure}
  \begin{figure}[]
 \caption{Three equivalent diagrams}\label{nonadjacent}
\begin{center}
\includegraphics{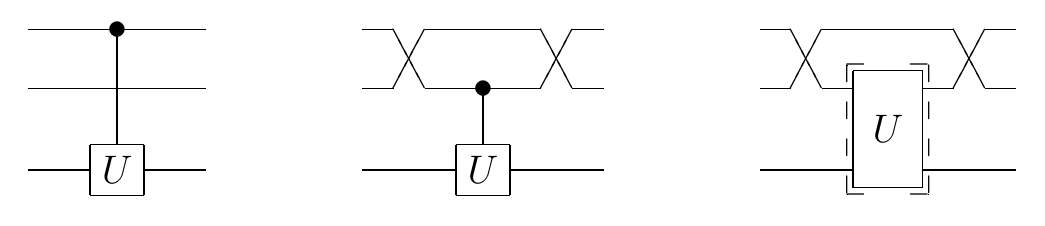}
\end{center}\end{figure}
  \begin{figure}[]
 \caption{The multiplicative symmetry via controlled additive symmetries}\label{multsym}
\begin{center}
\scalebox{0.8}{
\includegraphics{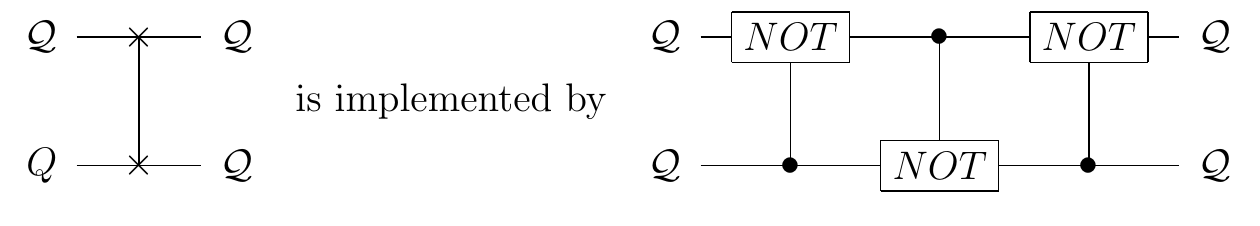}
}
\end{center}
\end{figure}

\subsection{Controlled Operations and Categorical Swap Maps}\label{notationalstuff}
In the standard quantum circuit formalism, controlled operations are not necessarily controlled by the qubit directly above them (i.e. the more significant qubit). We treat this as simply a diagrammatic convention, so a circuit  where the control qubit is not adjacent to the controlled operation is implemented using multiplicative symmetries (i.e. qubit swap maps) in the obvious way. Thus, the three circuits of Figure  \ref{nonadjacent} 
are equivalent, with the first being the usual quantum circuit notation, and the third conforming to the categorical conventions of Section \ref{strings}.

The qubit swap map (i.e. multiplicative symmetry) itself has an interesting categorical interpretation via the standard decomposition shown in Figure  \ref{multsym}. 
 The single qubit NOT gate ($NOT\ket{0}=\ket{1}$, $NOT\ket{1}=\ket{0}$) is the additive symmetry $s_{\mathbb C,\mathbb C}$ of two multiplicative unit objects. Figure \ref{multsym} 
expresses an abstract categorical identity relating the multiplicative symmetry $\sigma_{\mathbb 2,\mathbb 2}$, the additive symmetry $s_{I,I}$, and distributivity.  Details are left as an interesting exercise.

\subsection{Interpreting the Iterator 
in the Quantum Circuit Paradigm}
The interpretation of 
$!^{2^n}(U) \ = \ 1_H \oplus U \oplus U^2 \oplus \ldots \oplus U^{2^n-1}$
for some some unitary operation $U: H\rightarrow H$
is immediate; it is simply the sequence of multiply-controlled operations shown in Figure \ref{inefficientbang}. 
 \begin{figure}[]
 \caption{A circuit for $!^{2^n}(U)$}\label{inefficientbang}
\begin{center}
\scalebox{0.8}{
\includegraphics{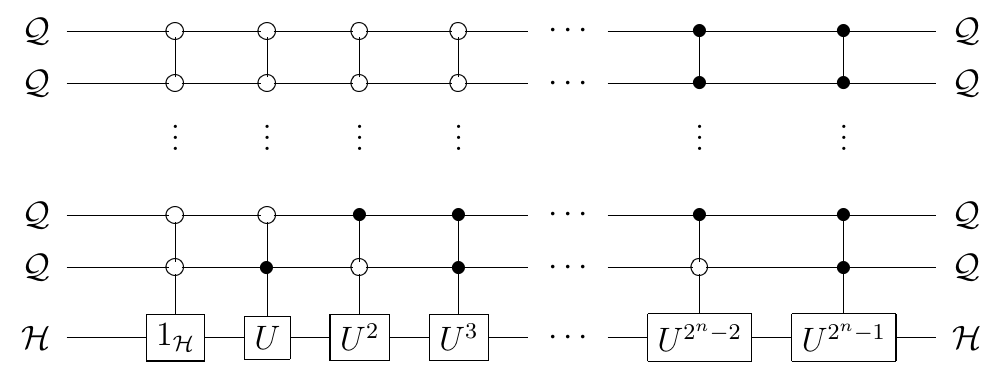}
}
\end{center}
\end{figure}
The operational interpretation is immediate:
\begin{proposition}\label{bangaction}
Given an arbitrary  quantum state $\ket{\psi}\in H$ and a computational basis ancilla state $\ket{a}$,  the circuit of Figure \ref{inefficientbang}
 acts on their tensor product as $\ket{a}\ket{\psi} \ \ \mapsto \ \ \ket{a}U^a\ket{\psi}$.
\end{proposition}
\begin{proof}
This follows by definition of the action of controlled operations in the quantum circuit model.
\qed
\end{proof}

\subsection{Applications of the $!^N(U)$ Operation in Quantum Programming}\label{periodfinding}
Quantum circuits acting as $\ket{a}\ket{\psi} \ \mapsto \ \ket{a}U^a\ket{\psi}$ have an important role to play in quantum period-finding algorithms (although, of course, the precise circuit of Figure 9 
is {\em not} used). 
The best-known period-finding algorithm is, of course, Shor's factorisation algorithm, based on period-finding for modular exponential functions. 

Period-finding algorithms rely on a central oracle that acts classically on some subset of the computational basis 
 (we refer to \cite{NC} for a formal definition, and \cite{TCS2} for a categorical interpretation in terms of Barr's $l_2$ functor \cite{MB}).  
Given a classical reversible function $f$, they require a unitary that acts as $\ket{a}\ket{x} \ \ \mapsto \ \ \ket{a}\ket{f^a(x)}$.
Given an oracle $U_f$ for the classical computation $f$, we may instead write this as $\ket{a}\ket{x} \ \ \mapsto \ \ \ket{a}U_f^a\ket{x}$
and observe that the required oracle for quantum period-finding is in fact $!^N(U_f)$, for some suitably large integer $N=2^n$. 

The complete quantum period-finding algorithm (up to some straightforward classical pre- and post- processing) is then simply given by conjugating such an oracle by a quantum Fourier transform, applied to the first register only.  For example, in Shor's algorithm the central oracle is required to implement modular exponentials, via the action $\ket{x}\ket{1} \ \ \mapsto \ \ \ket{x} \ket{r^x\ (mod\ K)}$
and hence, by linearity,
\[ \left( \sum_{x=0}^N \ket{x} \right) \ket{1} \ \ \mapsto \ \ \sum_{x=0}^N \ket{x} \ket{r^x\ (mod\ K)} \]
Given a (readily constructed) quantum oracle $U$ that acts on the computational basis  as $U\ket{p} \ = \ \ket{rp\ (mod\ K)}$, then, by Proposition \ref{bangaction}, the required oracle for Shor's algorithm  is $!^N(U)$. Thus the (quantum part of) Shor's algorithm is as shown in Figure \ref{bangshor}.


 \begin{figure}[]
 \caption{The quantum circuit in Shor's algorithm}\label{bangshor}
\begin{center}
\scalebox{0.8}{
\includegraphics{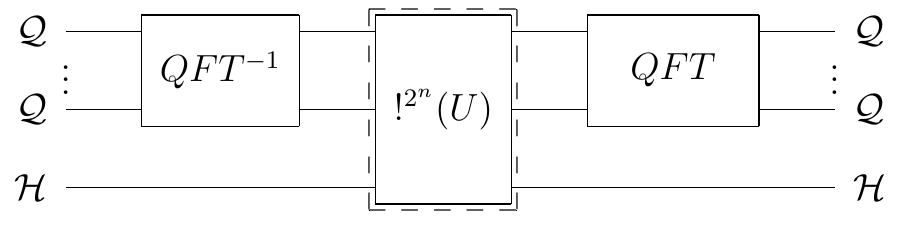}
}
\end{center}\end{figure}

\section{An Efficient Circuit for the $!^N(U)$ Operation}
The utility of the $!^N(U)$ operation in any period-finding algorithm must rely on an efficient implementation. Implementing the the central oracle using the circuit of Figure \ref{inefficientbang} would be pointless, 
given the complexity of constructing such a circuit. Instead, quantum algorithms (in particular, Shor's algorithm) use an exponentially more efficient  
circuit; we demonstrate that this is exactly the implementation given by Corollary \ref{efficientconstr}.

\begin{proposition}\label{AB}
The circuits $A$ and $B$ shown in Figure \ref{equivcircuits} 
are equivalent. 
\end{proposition}
\begin{figure}[]
\caption{Two equivalent quantum circuits}\label{equivcircuits}
\begin{center}
\scalebox{0.8}{
\includegraphics{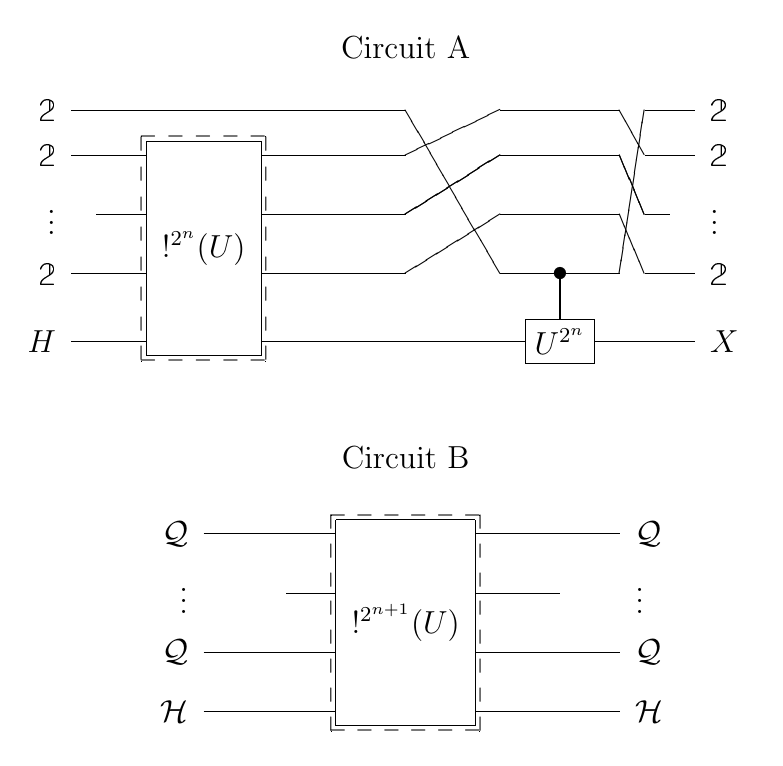}
}
\end{center}
\end{figure}
\begin{proof}
This follows directly from the Theorem \ref{maintheorem} -- in particular, the diagrammatic illustration given in Figure 2 makes it immediate. Note that the canonical multiplicative symmetry for qubits is drawn as a category-theoretic symmetry, rather than the traditional quantum circuit equivalent shown in Figure  \ref{multsym}. 
\qed
\end{proof}


A simple corollary of Proposition \ref{AB} above, and the notational simplifications of Section \ref{notationalstuff}, is that we may give a quantum circuit for $!^{2^n} (U)$ using $O(n)$ controlled quantum logic gates 
as follows:

\begin{corollary}
The circuit of Figure \ref{efficientbang} 
implements the $!^{2^n}(U)$ operation.
\end{corollary}
\begin{proof}This follows by induction on Proposition \ref{AB} above.\qed \end{proof}
\begin{figure}[]\caption{An efficient implementation of $!^{2^n}(U)$}\label{efficientbang}
\begin{center}
 \scalebox{0.8}{
\includegraphics{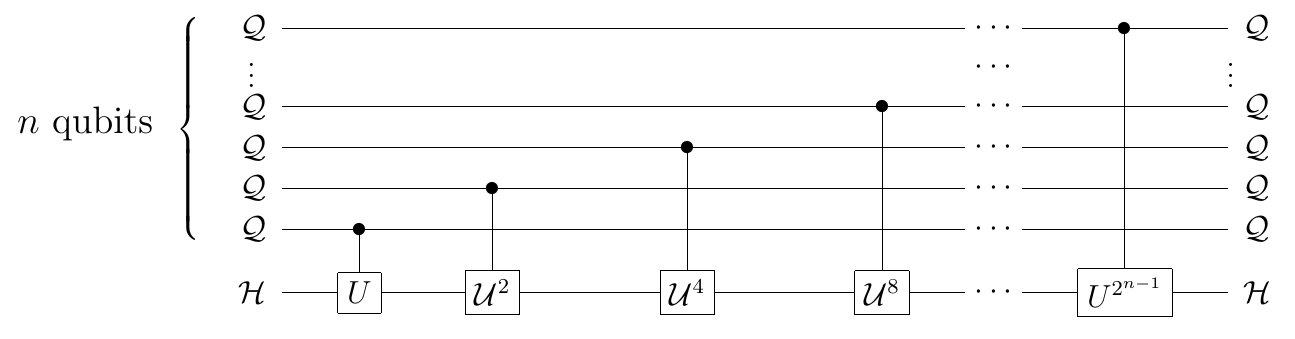}
}
\end{center}
\end{figure}

\begin{remark}
The efficient circuit of Figure  \ref{efficientbang} 
 is exactly the circuit used by P. Shor to implement modular exponentiation \cite{SH}. From a purely quantum circuit point of view, it is straightforward to demonstrate the equivalence of the circuits of Figure \ref{inefficientbang} 
 and Figure \ref{efficientbang}. The interest, from our point of view, is that this equivalence of circuits is an expression of a canonical coherence identity, and thus holds in any strongly distributive category. 
\end{remark}

\subsection{Oracles and Black Boxes}\label{blackboxes}
In referring to the circuit of Figure \ref{efficientbang} 
as requiring $O(n)$ primitive gates to implement $!^{2^n}(U)$, we have explicitly {\em not} considered the complexity of implementing $U,U^2,U^3,\ldots$. Rather, we have treated each of these operations as a `black box'. For concrete algorithms, this is a serious omission; in particular, any practical realisation of Shor's algorithm also requires some efficient way of implementing (controlled versions of) $U^{2^k}$, where the operation $U\ket{p} \ = \ \ket{rp\ (mod\ K)}$ is as described in Section \ref{periodfinding}.

Fortunately, such an efficient implementation also exists --- an oracle for the squaring operation $c\mapsto c^2\ (mod \ K)$ (up to some suitable ancilla, and garbage collection)  provides a simple, efficient way of implementing $U^{2^k}$, for $k=1,\ldots,n$. 
This is described in detail in \cite{SH}. Note that this technique is not available for arbitrary functions; rather, modular exponentiation is one of the few arithmetic functions for which such an efficient decomposition exists.

\section{Conclusions and Future Directions}\label{conclusions}
We have  demonstrated that the structural isomorphisms for strongly distributive categories have a role to play in understanding quantum algorithms -- or at least that perhaps familiar operations in quantum circuits can be given an abstract interpretation in terms of categorical coherence.

 Classically, equivalence up to canonical isomorphism is often used in program transformation, and it is pleasing, although not entirely unexpected, to see it in the quantum setting as well. Of more interest is how little of the machinery of categorical quantum mechanics we have used in establishing these transformations --- the only assumption required is that of two monoidal tensors related by distributivity up to isomorphism, and thus the constructions of this paper are valid in a wide range of different categorical and algebraic settings. 

\section*{Acknowledgements and Apologies}
\subsection*{Acknowledgements}
The author wishes to thank Robin Cockett for interesting discussions on the category theory of distributive categories, and applications. Similarly, thanks are due to Samson Abramsky and Bob Coecke for many discussions on the interpretation of distributivity and the direct sum, and on categorical quantum mechanics generally. Thanks are also due to Philip Scott, for discussions and references on categorical models of linear logic, with particular reference to the treatment of both distributivity and models of Girard's bang $!(\ )$ operation.

\subsection*{An Apology}
The work in this paper was first presented at a QICS quantum computing conference (Oxford 2010), under the title {\em `The role of coherence in quantum algorithms'} ({\em http://www.comlab.ox.ac.uk/quantum/content/1005021/}). After the talk, the speaker was approached by a delegation of experimental physicists, who explained that they had attended the talk because of the mention of `coherence' in the title, in the hope that it would be a break from the hard-core category theory presented in other talks. The author wishes to apologise for the (intentionally) misleading title, but hopes that they enjoyed the talk nevertheless.

\end{document}